\documentclass[11pt]{amsart}
\usepackage[latin1]{inputenc}
\usepackage{geometry}         
\usepackage{a4wide}       
\usepackage{graphicx}
\usepackage{xcolor}
\usepackage{amssymb,amsthm}
\usepackage{epstopdf}
\usepackage[pagebackref, colorlinks = true, linkcolor = teal, urlcolor  = teal, citecolor = violet]{hyperref}

\newtheorem{thm}{Theorem}
\newtheorem{proposition}{Proposition}
\newtheorem{lemma}{Lemma}
\newtheorem{definition}{Definition}

\title[QND: parameter estimation for mixtures of multinomials]{Quantum Non Demolition Measurements: \\parameter estimation for mixtures of multinomials}
\author{T. Benoist}
\address{Institut de Math\'ematiques de Toulouse UMR5219, Universit\'e de Toulouse ; CNRS, \'Equipe de Statistique et Probabilit\'es,
Universit\'e Paul Sabatier, 31062 Toulouse Cedex 9, France}
\email{tristan.benoist@math.univ-toulouse.fr}
\author{F.Gamboa}
\address{Institut de Math\'ematiques de Toulouse UMR5219, Universit\'e de Toulouse ; CNRS, \'Equipe de Statistique et Probabilit\'es,
Universit\'e Paul Sabatier, 31062 Toulouse Cedex 9, France}
\email{fabrice.gamboa@math.univ-toulouse.fr}
\author{C. Pellegrini}
\address{Institut de Math\'ematiques de Toulouse UMR5219, Universit\'e de Toulouse ; CNRS, \'Equipe de Statistique et Probabilit\'es,
Universit\'e Paul Sabatier, 31062 Toulouse Cedex 9, France}
\email{clement.pellegrini@math.univ-toulouse.fr}

\date{\today}                                          

\begin{document}
\begin{abstract}
In Quantum Non Demolition measurements, the sequence of observations is distributed as a mixture of multinomial random variables. Parameters of the dynamics are naturally encoded into this family of distributions. We show the local asymptotic mixed normality of the underlying statistical model and the consistency of the maximum likelihood estimator. Furthermore, we prove the asymptotic optimality of this estimator as it saturates the usual Cram\'er Rao bound.
\end{abstract}

\maketitle

\section{Introduction}
Measuring directly a small {\it quantum sized} physical system is done by letting it interact with a macroscopic instrument. This procedure can result in the destruction of the measured system. For example photons are absorbed to create an electronic signal. To avoid the destruction of the measured system, one relies then on indirect measurements. The system first interacts with an auxiliary system, or probe, that is then measured. The goal is then to infer the system state from the information obtained through this indirect measurement. Though, from the laws of quantum mechanics, this procedure induces a back action on the system that may change its state. Moreover, the measurement outcome being inherently random, the system state may become random itself. Hence, if one aims at measuring indirectly a physical quantity of the system, the indirect measurement must be tuned such that, if the system state corresponds to an almost sure value of the physical quantity of interest, then, indirectly measuring it will not modify its state. This kind of indirect measurements is called Quantum Non Demolition (QND) measurements. It has been  introduced in the eighties as a technique for precise measurements \cite{QND_origin}. Maybe one of the experiment illustrating best QND measurements is Haroche's group one. By sending atoms through a supra conducting cavity containing a monochromatic electromagnetic field, and measuring the atoms, it is possible to measure the number of photons inside the cavity whithout destroying them\cite{guerlin}.

To increase the amount of information on the system state obtained through a QND measurement, the procedure is repeated.
The system evolution is then described by an unobserved Markov chain $(\phi_n)$ (see Section \ref{sec:QND} for the complete description). The only handy observation is the sequence of the measurement results denoted by $(X_n)$.  
Baye's law maps the information of $(X_n)$ in the evolution of the Markov chain.

In general, the sequence  of random variables $(X_n)$ is not i.i.d. (even not Markovian). Therefore, statistical inference for QND measurement  cannot fully rely on standard results on  i.i.d. models.  Efficient parameter estimation is of course crucial for these experiments. Particularly if one hopes to have a faithful estimation of the system state Markov chain $(\phi_n)$.

In this paper, we show that the QND model enters perfectly in the framework of the usual statistical asymptotic theory. More precisely, we provide a complete study in terms of  local asymptotic mixed normality (LAMN) of the model (we refer to \cite{van} for the whole theory). It is worth noticing that a quantum analog of local asymptotic  normality (LAN) called quantum local asymptotic normality (QLAN) has been developed in the context of quantum statistics \cite{Guta1,Guta2}. Similarly quantum extensions of classical notions such as Quantum Fisher Information and Cram\'er Rao bound have been developed \cite[Section 2.2.5]{holevo}. Here, we will not follow this approach and will concentrate on more classical statistical properties.  Our results rely on the fact that our model, thanks to the QND condition, is actually a mixture of i.i.d. statistical models. More precisely, it has been show in \cite{bbb1,bbb2,bbb3,bp,jf} that the probability space describing these experiments can be divided into asymptotic events (belonging to the tail algebra) such that $(X_n)$ conditioned to one of such asymptotic event is a sequence of i.i.d. random variables. So that, the law of $(X_n)$ is a mixture of i.i.d. laws. The weights involved into the mixture depends on the initial state of the system.

The conditioning making $(X_n)$ an i.i.d. sequence is highly exploited in order to derive the LAMN property. To our knowledge this is the first time that the LAMN property is shown in this context.
After proving the LAMN property we study the maximum likelihood estimation and prove that it is optimal in the sense that the Cram\'er Rao bound is achieved asymptotically. Note that parametric estimation for indirect measurements has been previously investigated in \cite{Guta3,Guta4,Rouchon} with different assumptions. But, the theory of asymptotic likelihood has not been studied therein. 
\smallskip

The paper is organised as follows. In Section \ref{sec:multi}, we discuss the model of multinomial mixture studied along the paper. In Section \ref{sec:LAMN} we show the local asymptotic mixed normality. Section \ref{sec:maxlike} is devoted to the results for the maximum likelihood estimator (consistency and saturation of the Cram\'er Rao bound).  Finally in Section \ref{sec:QND} we work on the QND model underlining the link with multinomial mixtures. Further, some numerical simulations illustrate our results on a QND toy model inspired by \cite{guerlin}.

\section{Mixture of multinomials}\label{sec:multi}
Let $\mathcal A=\{1,\ldots,l\}$, $\mathcal P=\{1,\ldots,d\}$ and $\Theta$ be a compact subset of $\mathbb R^D$ with a non empty interior. For any $\alpha\in\mathcal P$ and $\theta\in\Theta$, the quantity
$(p_\theta(j|\alpha))_{j\in\mathcal A}$
denotes a probability distribution over $\mathcal A$, that is $p_\theta(j\vert\alpha)\in(0,1)$ and $\sum_{j\in\mathcal A}p_\theta(j\vert\alpha)=1$. In the sequel for any $j\in\mathcal A$ and $\alpha\in\mathcal P$ the notation $p_{.}(j\vert\alpha)$ holds for the function $\theta\mapsto p_\theta(j\vert\alpha)$. 

Let $\mathbb P$ be a probability distribution, we will use the notation $\stackrel{\mathcal L-\mathbb P}=$ to mean equality in distribution when the underlying probability space is endowed with $\mathbb P$. This notation will also be used for convergence in distribution  writing sometimes  $\stackrel{\mathcal L-\mathbb P_n}=$ when a family of probability measures $(\mathbb P_n)$ is involved. In the paper the notation $\mathcal N(m,\sigma^2)$ is used for the Gaussian distribution of mean $m$ and variance $\sigma^2$. For any $x\in\mathbb R^D$ and any subset $A\subset\mathbb R^D,$ the set $A+x$ will denote $A+x=\{y+x,y\in A\}$.

Let us now describe the probability model that we will study. Let $\Omega=\mathcal A^{\mathbb N}$ and let $\mathcal F$ be the smallest $\sigma$-algebra containing the cylinder sets $\{\omega\in\Omega|\omega_k=j_k, \forall k\leq n\}$. All the measures introduced afterwards are defined on the measurable space $(\Omega,\mathcal F)$ without mentioning it.

For any $\theta\in\Theta$ and for each $\alpha\in\mathcal P$, let $\mathbb P_{\theta|\alpha}$ be the multinomial probability measure built on the weights $(p_\theta(j|\alpha))_{j\in\mathcal A}$. Namely, for any $n$-tuple $(j_1,\ldots,j_n)\in\mathcal A^n$,
$$\mathbb P_{\theta|\alpha}(j_1,j_2,\ldots,j_n)=\prod_{k=1}^n p_\theta(j_k|\alpha).$$
Let $(q(\alpha))_{\alpha\in\mathcal P}$ be a probability measure on $\mathcal P$, we denote by $\mathbb P^q_\theta$ the probability measure defined as a convex combination of the measure $\mathbb P_{\theta|\alpha}$ with weights $(q(\alpha))_{\alpha\in\mathcal P}$:
$$\mathbb P^q_\theta=\sum_{\alpha\in\mathcal P}q(\alpha)\mathbb P_{\theta|\alpha}$$
Without loss of generality, we shall always assume that $q(\alpha)>0$ for all $\alpha\in\mathcal P$. Indeed, one can reduce the set $\mathcal P$ if needed.

We study the statistical model $(\mathbb P^q_\theta, \theta\in\Theta)$. As we shall see, under the identifiability condition below, our results will be independent of $q$. Hence, in the sequel we shall alleviate the notation replacing $\mathbb P_\theta^q$ by $\mathbb P_\theta$.
\medskip

Our main assumption on the different multinomials is the following: 

\textbf{Assumption ID:} For any $(\alpha,\theta)$ and $(\beta,\theta')\in\mathcal P\times\Theta$ such that $(\alpha,\theta)\neq (\beta,\theta')$, $\exists j\in\mathcal A$ such that,
$$p_{\theta}(j|\alpha)\neq p_{\theta'}(j|\beta).$$

\textbf{Remark:} If for every $\theta\in\Theta$, Assumption {\bf ID} does not hold for two couples $(\alpha,\theta)$ and $(\beta,\theta)$, it can be enforced by reducing $\mathcal P$, identifying the elements giving the same distributions $p_\theta(\cdot|\alpha/\beta)$.

From now on, we will denote by $\theta^*\in \Theta$ the true value of the parameter $\theta$. We assume that $\theta^*$ is in the interior of $\Theta$. 
The following definition will be usefull.

\begin{definition} 
For any $\theta^\ast\in\Theta$ and $\gamma\in\mathcal P$, let,
$$\Omega_{\theta^*|\gamma}:=\operatorname{supp} \mathbb P_{\theta^*|\gamma}\quad\text{and}\quad \Gamma_{\theta^*}:=\sum_{\gamma\in\mathcal P} \gamma\mathbf{1}_{\Omega_{\theta^*|\gamma}}.$$
\end{definition}
When it does not lead to confusion we may omit the index $\theta^*$ for both the sets $\Omega_{\theta^*|\gamma}$ and the random variable $\Gamma_{\theta^*}$.

\noindent\textbf{Remark:} It is a direct consequence of Assumption {\bf ID} that 
$$\mathbb P_{\theta^*}(\Omega_{\theta^*|\gamma})=\mathbb P_{\theta^*}(\Gamma_{\theta^*}=\gamma)=q(\gamma).$$

At this stage we need to introduce some quantities  quantifying the information and proximity in our models. In particular, we shall use many times the Shannon entropy given a parameter $\theta$ and the Kullback--Leibler divergence given $\theta$ with respect to $\theta'$. For $\alpha,\beta\in\mathcal P$ and $\theta,\theta'\in\Theta$, let
\begin{equation}\label{eq:Sha}S_{\theta}(\alpha):=-\sum_{j\in\mathcal A}p_\theta(j|\alpha)\ln p_\theta(j|\alpha),\end{equation}
be the Shannon entropy and
\begin{equation}\label{eq:KLd}S_{\theta|\theta'}(\alpha|\beta):=\sum_{j\in\mathcal A}p_\theta(j|\alpha)\big(\ln p_\theta(j|\alpha)-\ln p_{\theta'}(j|\beta)\big),
\end{equation}
be the Kullback--Leibler divergence. In Eq \eqref{eq:KLd}, when $\theta=\theta'$ we just write $S_\theta$, that is for $\alpha,\gamma\in\mathcal P$ and $\theta\in\Theta$ 
\begin{equation}\label{eq:relat}S_{\theta}(\alpha|\gamma)=S_{\theta|\theta}(\alpha|\gamma)=\sum_{j\in\mathcal A}p_\theta(j\vert\alpha)\ln\frac{p_{\theta}(j|\alpha)}{p_{\theta}(j|\gamma)}.
\end{equation}
\noindent\textbf{Remark:} Note that assumption {\bf ID} ensures that \eqref{eq:KLd} and \eqref{eq:relat} are strictly positive.

We now state a technical lemma which is a key tool to all our proofs (the arguments are closely related to the one used in \cite{bbb2}).
\begin{lemma}\label{lem:conditionning_on_tail}
Assume that {\bf ID} holds.
\begin{enumerate}
\item {\bf Almost sure convergence} Let $(X_n^\gamma)$, be a sequence of random variables depending on $\gamma\in\mathcal P$. If for any $\gamma\in\mathcal P$
$$\lim_n X_n^\gamma=X^\gamma,\quad \mathbb P_{\theta^*|\gamma}-a.s.$$
Then
$$\lim_n X_n^{\Gamma_{\theta^*}}=X^{\Gamma_{\theta^*}},\quad \mathbb P_{\theta}-a.s$$

\item {\bf Convergence in distribution} Let $(\theta_n^*)$ be a sequence in $\Theta$ and $(X_n^\gamma)$ as in $(1)$. If for any $\gamma\in\mathcal P$

 $$\lim_n X_n^\gamma\,\,\stackrel{\mathcal L-\mathbb{P}_{\theta^*_n|\gamma}}=\,\, X^\gamma$$
  Then, $$\lim_n X_n^{\Gamma_{\theta^*}}\stackrel{\mathcal L-\mathbb{P}_{\theta^*_n}}=X^\Gamma,$$
  where $\Gamma$ is a r.v. whose distribution is given by $\operatorname{Pr}(\Gamma=\gamma)=q(\gamma)$.
\end{enumerate}
\end{lemma}
\begin{proof}
Assume that $(X_n^\gamma)$ converges almost surely towards $X^\gamma$ w.r.t $\mathbb P_{\theta^*|\gamma}$.  Then, $$\mathbb P_{\theta^*|\gamma}(\cap_{N}\cup_{n_0}\cap_{n\geq n_0}\{|X_n^\gamma-X^\gamma|< 1/N\})=1.$$
Since $\Gamma=\gamma$ $\mathbb P_{\theta^*|\gamma}$-a.s.,
$$\mathbb P_{\theta^*|\gamma}(\cap_{N}\cup_{n_0}\cap_{n\geq n_0}\{|X_n^\Gamma-X^\Gamma|< 1/N\})=1.$$
This is true for any $\gamma\in\mathcal P$. Since $\mathbb P_{\theta^*}$ is convex combination of the measures $\mathbb P_{\theta^*|\gamma}$,
$$\mathbb P_{\theta^*}(\cap_{N}\cup_{n_0}\cap_{n\geq n_0}\{|X_n^\Gamma-X^\Gamma|< 1/N\})=1$$
and (1) holds.

\medskip
Assume that $(X_n^\gamma)$ converges weakly towards $X^\gamma$ w.r.t. $(\mathbb P_{\theta_n^*|\gamma})$. Then,
$$\lim_{n\to\infty} \mathbb E_{\theta^*_n|\gamma}(f(X_n^\gamma))=\mathbb E(f(X^\gamma))$$
for any continuous and bounded function $f$. Since $\mathbb P_{\theta^*_n|\gamma}(\Gamma_{\theta^*_n}=\gamma)=1$ and $\mathbb P_{\theta^*_n}=\sum_{\alpha\in\mathcal P}q(\alpha)\mathbb P_{\theta^*_n|\alpha}$,
$$\lim_{n\to\infty} \mathbb E_{\theta^*_n}(f(X_n^{\Gamma_{\theta^*}}))=\sum_{\alpha\in\mathcal P} q(\alpha)\mathbb E(f(X^\alpha)).$$
That convergence yields (2).
\end{proof}


In the sequel we shall use the process $(N_n(j))_{j\in\mathcal A}$ where for all $n\in\mathbb N$ and all $j\in\mathcal A$
$$N_n(j)(\omega)=\sum_{k=1}^n\mathbf 1_{\omega_k=j},$$
for all $\omega\in\Omega$, which counts the number of times the result $j$ appears before time $n$. Remark that $\sum_{j\in\mathcal A}N_n(j)=n$. The strong law of large number for i.i.d $L^1$ random variables involves that
\begin{equation}\label{eq:LFGN}
\lim_{n\rightarrow\infty}\frac{N_n(j)}{n}=p_\theta(j\vert\alpha),\quad \mathbb P_{\theta|\alpha}-a.s., 
\end{equation} 
for all $\alpha\in\mathcal P$, all $j\in\mathcal A$ and all $\theta\in\Theta.$

\section{Local Asymptotic Mixed Normality}\label{sec:LAMN}

We first prove that the statistical model $(\mathbb P_\theta,\theta\in\Theta)$ is asymptotically equivalent to a mixture of Gaussian models. Let us first recall the definition of local mixed asymptotic normality that we shall prove for our model.
\begin{definition}[Local Asymptotic Mixed Normality(LAMN)\cite{van}]
A sequence of statistical models $(\mathbb Q_{\theta}|_{\mathcal F_n}|\theta\in\Theta)$ is said to be LAMN at $\theta^*$ if, there exits $\Delta_{\theta^*}$ and $J_{\theta^*}$ two random variables such that for any $h\in\Theta-\theta^*$
$$\lim_{n\to\infty}\ln\frac{\mathbb Q_{\theta^*+h/\sqrt{n}}(\omega_1,\ldots,\omega_n)}{\mathbb Q_{\theta^*}(\omega_1,\ldots,\omega_n)}\stackrel{\mathcal{L}-\mathbb Q_{\theta^*}}=h^\mathsf{T}\Delta_{\theta^*}-\frac12h^\mathsf{T}J_{\theta^*}h,$$
 where the law of $\Delta_{\theta^*}$ conditioned on $J_{\theta^*}=J$ is $\mathcal N(0,J)$.
\end{definition}
The LAMN property for our model follows from next lemma. It reduces the problem to local asymptotic normality for i.i.d. multinomials.
\begin{lemma}\label{lem:convergence_P_to_p_gamma_in_log}
 Assume that {\bf ID} holds and that for any $(j,\alpha)\in{\mathcal A}\times{\mathcal P}$ the function $p_{.}(j|\alpha)$ is continuous in $\theta^\ast$. Let $(\theta_n)\in\Theta$ be a sequence of random variables such that $\lim_n\theta_n=\theta^*, (\mathbb P_{\theta^*|\gamma}-a.s)$. Then,
\begin{equation}
\left|\frac{\mathbb P_{\theta_n}(\omega_1,\ldots,\omega_n)}{q(\gamma)\mathbb P_{\theta_n|\gamma}(\omega_1,\ldots,\omega_n)}-1\right|=o(n^{-1/2}),\quad \mathbb P_{\theta^*|\gamma}-a.s.
\end{equation}
\end{lemma}
\begin{proof}
We prove the stronger result that the convergence is exponential but we will only need the convergence at order $o(n^{-1/2})$.

From the definition of $\mathbb P_{\theta_n}$,
$$\frac{\mathbb P_{\theta_n}(\omega_1,\ldots,\omega_n)}{\prod_{k=1}^np_{\theta_n}(\omega_k|\gamma)}=q(\gamma)+\sum_{\alpha\neq \gamma}q(\alpha)\exp\left(\sum_{j\in\mathcal A}N_n(j)(\omega)\ln\frac{p_{\theta_n}(j|\alpha)}{p_{\theta_n}(j|\gamma)}\right).$$
The strong law of large numbers \eqref{eq:LFGN} and the continuity assumption on the functions $p_{.}(j|\alpha)$ imply
$$\lim_{n\to\infty}\frac1n \sum_{j\in\mathcal A}N_n(j)\ln\frac{p_{\theta_n}(j|\alpha)}{p_{\theta_n}(j|\gamma)}=-S_{\theta^\ast}(\gamma|\alpha),\quad \mathbb P_{\theta^*|\gamma}-a.s.$$
Recall that Assumption {\bf ID} implies $S_{\theta^*}(\gamma|\alpha)>0$ for any $\alpha\neq \gamma$. Therefore, there exists $s>0$ independent of $\gamma$, $\alpha$ and of the sequence $(\theta_n)$ such that,
$$\lim_{n\to\infty}e^{n s}\sum_{\alpha\neq \gamma}q(\alpha)\exp\left(\sum_{j\in\mathcal A}N_n(j)\ln\frac{p_{\theta_n}(j|\alpha)}{p_{\theta_n}(j|\gamma)}\right)=0,\quad \mathbb P_{\theta^*|\gamma}-a.s.$$

Now, 
$$\frac{\mathbb P_{\theta_n}(\omega_1,\ldots,\omega_n)}{q(\gamma)\mathbb P_{\theta_n|\gamma}(\omega_1,\ldots,\omega_n)}-1=\frac{1}{q(\gamma)}\sum_{\alpha\neq \gamma}q(\alpha)\exp\left(\sum_{j\in\mathcal A}N_n(j)\ln\frac{p_{\theta_n}(j|\alpha)}{p_{\theta_n}(j|\gamma)}\right),$$
so that the result follows.
\end{proof}

Now we state our main result.

\begin{thm}[LAMN]
 Assume that {\bf ID} holds and that for any $(j,\alpha)\in\mathcal A\times\mathcal P$ the function $ p_{.}(j|\alpha)$ is differentiable at $\theta^*$. Then, the sequence $(\mathbb P_\theta|_{\mathcal F_n}|\theta\in\Theta)$ is LAMN in $\theta^*$ with
$$J_{\theta^*}:=\sum_{j\in\mathcal A}p_{\theta^*}(j|\Gamma)\Big(\nabla_\theta \ln p_\theta(j|\Gamma)(\nabla_\theta \ln p_\theta(j|\Gamma))^\mathsf{T}\Big)_{|_{\theta=\theta^*}}=:I_{\theta^*}(\Gamma).$$
\end{thm}
\begin{proof}

Fix $\gamma\in\mathcal P$. From standard results on parameter estimation for i.i.d. models, the sequence of models $(\mathbb P_{\theta|\gamma}|_{\mathcal F_n}|\theta\in\Theta)$ is locally asymptotic normal in $\theta^*$ (see \cite{van}). Namely, there exists a r.v. $Z^\gamma\sim\mathcal N(0,I_{\theta^*}(\gamma))$ such that, for any $h\in\Theta-\theta^*$,
$$\lim_{n\to\infty} \sum_{k=1}^n \ln\frac{p_{\theta^*+\frac{h}{\sqrt{n}}}(\omega_k|\gamma)}{p_{\theta^*}(\omega_k|\gamma)}\stackrel{\mathcal L-\mathbb P_{\theta^*|\gamma}}=h^\mathsf{T}Z^\gamma -\frac12h^\mathsf{T}I_{\theta^*}(\gamma)h.$$

From Lemma \ref{lem:convergence_P_to_p_gamma_in_log} setting $\theta_n=\theta^*$ and $\theta_n=\theta^*+h/\sqrt{n}$,
$$\lim_{n\to\infty}\ln\frac{\mathbb P_{\theta^*+\frac{h}{\sqrt{n}}}(\omega_1,\ldots,\omega_n)}{\mathbb P_{\theta^*}(\omega_1,\ldots,\omega_n)}\ - \sum_{k=1}^n \ln\frac{p_{\theta^*+\frac{h}{\sqrt{n}}}(\omega_k|\gamma)}{p_{\theta^*}(\omega_k|\gamma)}=0, \quad \mathbb P_{\theta^*|\gamma}-a.s.$$
It follows that,
$$\lim_{n\to\infty}\ln\frac{\mathbb P_{\theta^*+\frac{h}{\sqrt{n}}}(\omega_1,\ldots,\omega_n)}{\mathbb P_{\theta^*}(\omega_1,\ldots,\omega_n)}\stackrel{\mathcal L-\mathbb P_{\theta^*|\gamma}}=h^\mathsf{T}Z^\gamma -\frac12h^\mathsf{T}I_{\theta^*}(\gamma)h,$$
so that Lemma \ref{lem:conditionning_on_tail} yields the Proposition.
\end{proof}

\section{Maximum likelihood estimation}\label{sec:maxlike}

This section is devoted to the study of the maximum likelihood estimator. For $\omega\in\Omega$ at step $n$ the log likelihood is defined as
$$\ell_n(\theta)(\omega_1,\ldots,\omega_n)=\frac1n\ln\mathbb P_{\theta}(\omega_1,\ldots,\omega_n).$$
We study the maximum likelihood estimator defined as
$$\hat\theta_n:=\operatorname{argmax}_{\theta\in\Theta}\ell_n(\theta).$$

\subsection{Consistency}
We define the function
\begin{equation}
\ell_{\theta^*,\gamma}(\theta)=-S_{\theta^*}(\gamma)-\min_{\alpha\in\mathcal P}S_{\theta^*|\theta}(\gamma|\alpha),
\end{equation}
for all $\gamma\in\mathcal P$ and all $\theta\in\Theta$.

The following Lemma shows the almost sure uniform convergence of the sequence of log likelihood functions.
\begin{lemma}\label{lem:uniform_conv_ell}
Assume that {\bf ID} holds and that for any couple $(j,\alpha)\in\mathcal A\times \mathcal P$ the function $p_{.}(j|\alpha)$ is continuous. Then for any $\gamma\in\mathcal P$,
$$\lim_{n\to\infty} \sup_{\theta\in\Theta}|\ell_n(\theta)- \ell_{\theta^*,\gamma}(\theta)|=0, \quad \mathbb P_{\theta^*|\gamma}-a.s.,$$
and $\ell_{\theta^*,\gamma}$ is such that, for any $\epsilon>0$
\begin{equation}\label{eq:theta_global_max}
\ell_{\theta^*,\gamma}(\theta^*)>\sup_{\theta: d(\theta,\theta^*)\geq \epsilon} \ell_{\theta^*,\gamma}(\theta).
\end{equation}
\end{lemma}
\begin{proof}
Since $\Theta$ is compact and that for any $(j,\alpha)\in\mathcal A\times\mathcal P$ the function $p_{.}(j|\alpha)$ is continuous and positive, we get
$$\max_{(j,\alpha)\in\mathcal A\times\mathcal P}\sup_{\theta\in\Theta}|\ln p_{\theta}(j|\alpha)|<\infty.$$

Furthermore, the strong law of large numbers \eqref{eq:LFGN} implies that for any $\alpha\in\mathcal P$,
$$\lim_{n\to\infty}\sup_{\theta\in\Theta}\left|\sum_{j\in\mathcal A}\frac{N_n(j)}{n}\ln\left(\frac{p_{\theta}(j|\alpha)}{p_{\theta^*}(j|\gamma)}\right)+S_{\theta^*|\theta}(\gamma|\alpha)\right|=0,\quad \mathbb P_{\theta^*|\gamma}-a.s.$$
Hence,
$$\lim_{n\to\infty}\sup_{\theta\in\Theta}\left|\frac1n \ln \frac{\mathbb P_{\theta|\alpha}(\omega_1,\ldots, \omega_n)}{\mathbb P_{\theta^*|\gamma}(\omega_1,\ldots, \omega_n)}+S_{\theta^*|\theta}(\gamma|\alpha)\right|=0,\quad \mathbb P_{\theta^*|\gamma}-a.s.$$

It follows then from a repeated application of Lemma \ref{lem:log_sum_unfi_conv} (its statement and proof are postponed to the Appendix), that
$$\lim_{n\to\infty}\sup_{\theta\in\Theta}\left|\ell_n(\theta)- \ell_{\theta^*,\gamma}(\theta)\right|=0, \quad \mathbb P_{\theta^*|\gamma}-a.s.$$

Since $\ell_{\theta^*,\gamma}(\theta)=-S_{\theta^*}(\gamma)-\min_{\alpha\in\mathcal P}S_{\theta^*|\theta}(\gamma|\alpha)$, Assumption {\bf ID} implies the inequality \eqref{eq:theta_global_max} and the lemma is proved.
\end{proof}

The consistency of the maximum likelihood estimator  follows now from standard arguments.
\begin{proposition}\label{prop:consistent}
Assume that Assumption {\bf ID} holds and that for any $(j,\alpha)\in\mathcal A\times\mathcal P$ the function $ p_.(j|\alpha)$ is continuous. Then,
$$\lim_{n\to\infty}\hat\theta_n=\theta^*,\quad\mathbb P_{\theta^*}-a.s.$$
\end{proposition}
\begin{proof}
From Lemma \ref{lem:uniform_conv_ell}, classical results of parametric estimation (see \cite[Theorem 5.7]{van}) give,
$$\lim_{n\to\infty} \hat \theta_n=\theta^*,\quad \mathbb P_{\theta^*|\gamma}-a.s.$$
So that,  Lemma \ref{lem:conditionning_on_tail} yields the consistency w.r.t. $\mathbb P_{\theta^*}$.
\end{proof}

\subsection{Saturation of Cramér--Rao bound}
It is well know that if a model is LAMN, it verifies an asymptotic Cramér-Rao bound. Namely if $(T_n)$ is a sequence of estimators  such that $$\lim_{n\to\infty}\sqrt{n}(T_n- (\theta^*+h/\sqrt{n}))\stackrel{\mathcal L-\mathbb P_{\theta^*+h/\sqrt{n}}}=T$$ for any $h\in\Theta-\theta^*$, then, $\mathbb E(T\ T^{\mathsf{T}})- \sum_{\alpha\in\mathcal P} q(\alpha)I_\theta^{-1}(\alpha)$ is positive semi definite \cite[Corollary 9.9]{van}.

We now prove that the maximum likelihood estimator saturates this asymptotic bound. We prove it comparing $\hat \theta_n$ with $\hat \theta_n^\gamma$ defined by
$$ \hat \theta_n^\gamma:=\operatorname{argmax}_{\theta\in\Theta}\ell_n^\gamma(\theta),$$
where
$$\ell_n^\gamma(\theta)(\omega_1,\ldots,\omega_n)= \frac1n \ln \mathbb P_{\theta|\gamma}(\omega_1,\ldots,\omega_n)$$
for all $\theta\in\Theta$ and all $\omega\in\Omega$.
\begin{lemma}\label{lem:convergence_two_estimator}
 Assume that {\bf ID} holds and that for any $(j,\alpha)\in\mathcal A\times \mathcal P$  the function $p_.(j|\alpha)$ is twice continuously differentiable in a neighborhood of $\theta^*$.
Assume further that for each $\gamma\in\mathcal P$, $I_{\theta^*}(\gamma)$ is not singular.
Then,
$$\lim_{n\to\infty} \sqrt{n}(\hat \theta_n-\hat \theta_n^\gamma)=0,\quad \mathbb P_{\theta^*|\gamma}-a.s.$$
\end{lemma}
\begin{proof}
Standard results of parametric estimation for i.i.d. random variables and our assumption {\bf ID} imply $\lim_{n\to\infty}\hat \theta_n^\gamma=\theta^*$ $\mathbb P_{\theta^*|\gamma}$-a.s.
From the definition of the maximum likelihood estimators, 
\begin{equation}\label{eq:equality_diff_estimator}
\nabla \ell_n^\gamma(\hat \theta^\gamma_n)-\nabla \ell_n^\gamma(\hat \theta_n)=\nabla\ell_n(\hat \theta_n)-\nabla \ell_n^\gamma(\hat \theta_n)
\end{equation}
Let $\mathcal O$ be a sufficiently small neighborhood of $\theta^\ast$ (on which $p_.$ is regular). The consistency of $\hat \theta_n$ and $\hat \theta_n^\gamma$ ensures that $\mathbb P_{\theta^*|\gamma}$-a.s., for $n$ large enough, the likelihood estimators $\hat \theta_n$ and $\hat\theta_n^\gamma$ belong to $\mathcal O$. Hence, since by assumption $\ell_n^\gamma$ is twice differentiable in a neighborhood of $\theta^*$, there exists a sequence of random variable $\xi_n$ lying in the segment with extremity $\hat\theta_n$ and $\hat \theta_n^\gamma$ such that $\mathbb P_{\theta^*|\gamma}-a.s.$, for $n$ large enough,
$$\nabla^2\ell_n^{\gamma}(\xi_n)(\hat \theta_n-\hat \theta_n^\gamma)=\nabla \ell_n^\gamma(\hat \theta_n)-\nabla\ell_n(\hat \theta_n).$$
The last equality comes from \eqref{eq:equality_diff_estimator} and the Mean Value Theorem.

Now note that using explicit derivation and the strong law of large number \eqref{eq:LFGN}, we have  $$\lim_{n\to\infty}\nabla^2\ell_n^{\gamma}(\xi_n)=I_{\theta^*}(\gamma), \mathbb P_{\theta^*|\gamma}-a.s.$$ Since $I_{\theta^*}(\gamma)$ is assumed to be non singular and $\theta\mapsto I_{\theta}(\gamma)$ is continuous in a neighborhood of $\theta^*$, the previous convergence implies that for $n$ large enough $\nabla^2\ell_n^{\gamma}(\xi_n)$ is invertible. It follows that for $n$ large enough,
\begin{equation}\label{eq:egal_diff_theta_ll__proba}
|\hat \theta_n-\hat \theta_n^\gamma|=|[\nabla^2\ell_n^{\gamma}(\xi_n)]^{-1}(\nabla \ell_n^\gamma(\hat \theta_n)-\nabla\ell_n(\hat \theta_n))|\quad \mathbb P_{\theta^*|\gamma}-a.s.
\end{equation}
Let us now prove that $\nabla \ell_n^\gamma(\hat \theta_n)-\nabla\ell_n(\hat \theta_n))=o(n^{-1/2})$ $\mathbb P_{\theta^*|\gamma}$-a.s.

Explicit differentiation leads to,
$$\nabla \ell_n^\gamma(\hat \theta_n)-\nabla\ell_n(\hat \theta_n)=\sum_{\alpha\in\mathcal P} \sum_{j\in\mathcal A}\frac{ N_n(j)}{n}\left(\frac{\nabla p_{\hat \theta_n}(j|\gamma)}{p_{\hat \theta_n}(j|\gamma)}-\frac{\nabla p_{\hat \theta_n}(j|\alpha)}{p_{\hat \theta_n}(j|\alpha)}\right) \frac{q(\alpha)\mathbb P_{\hat\theta_n|\alpha}(\omega_1,\ldots,\omega_n)}{\mathbb P_{\hat\theta_n}(\omega_1,\ldots,\omega_n)}.$$
From the consistency of $\hat \theta_n$ and the continuous differentiability of $ p_{.}(j|\alpha)$ in a neighborhood of $\theta^*$, it follows that, there exists $C>0$ such that for $n$ large enough,
$$|\nabla \ell_n^\gamma(\hat \theta_n)-\nabla\ell_n(\hat \theta_n))|\leq C\left(1-\frac{q(\gamma)\mathbb P_{\hat \theta_n|\gamma}(\omega_1,\ldots,\omega_n)}{\mathbb P_{\hat\theta_n}(\omega_1,\ldots,\omega_n)}\right)\quad \mathbb P_{\theta^*|\gamma}-a.s.$$
Equation \eqref{eq:egal_diff_theta_ll__proba} combined with the $\mathbb P_{\theta^*|\gamma}$-a.s. convergence $\lim_{n\to\infty}\nabla^2\ell_n^{\gamma}(\xi_n)=I_{\theta^*}(\gamma)$, and with the Lemma \ref{lem:convergence_P_to_p_gamma_in_log} yields the result.
\end{proof}
\begin{proposition}
 Assume that {\bf ID} holds and that for any $(j,\alpha)\in\mathcal A\times \mathcal P$ the function $ p_{.}(j|\alpha)$ is three times continuously differentiable in a neighborhood of $\theta^*$. Assume further that for any $\gamma\in\mathcal P$, $I_{\theta^*}(\gamma)$ is not singular. 
Let $Z\sim\mathcal N(0,I_d)$ and $\Gamma$ be a random variable independent of $Z$ and taking value in $\mathcal P$ such that $\operatorname{Pr}(\Gamma=\gamma)=q(\gamma)$.Then, for any $h\in\Theta-\theta^*$, 
$$\lim_{n\to\infty}\sqrt{n}(\hat{\theta}_n-(\theta^*+{h}/{\sqrt{n}}))\stackrel{\mathcal L-\mathbb P_{\theta^*+h/\sqrt{n}}}=I_{\theta^*}(\Gamma)^{-\frac12}Z.$$
\end{proposition}
\begin{proof}
From standard results in parameter estimation\cite{van}, under Assumption {\bf ID} the statistical model $(\mathbb P_{\theta|\gamma},\theta\in\Theta)$ is LAN. Moreover, for any $h\in\Theta-\theta^*$, under the assumptions of the proposition,
\begin{equation}\label{eq:CLT_theta_gamma}
\lim_{n\to\infty}\sqrt{n}(\hat \theta_n^\gamma -(\theta^*+h/\sqrt{n}))\stackrel{\mathcal L-\mathbb P_{\theta^*|\gamma}}=I_{\theta^*}(\gamma)^{-\frac12} Z-h.
\end{equation}
Actually, the proof of $(\mathbb P_{\theta|\gamma},\theta\in\Theta)$ being LAN and the weak convergence \eqref{eq:CLT_theta_gamma} are based on the same Central Limit Theorem. It follows that,
\begin{equation*}
\begin{split}
\lim_{n\to\infty} \left(\sqrt{n}(\hat \theta_n^\gamma-(\theta^*+h/\sqrt{n})),\ \ln\frac{\mathbb P_{\theta^*+h/\sqrt{n}|\gamma}(\omega_1,\ldots,\omega_n)}{\mathbb P_{\theta^*|\gamma}(\omega_1,\ldots,\omega_n)}\right)\\
\stackrel{\mathcal L-\mathbb P_{\theta^*|\gamma}}=\left(I_{\theta^*}^{-\frac12}(\gamma)Z-h,\ h^{\mathsf{T}}I_{\theta^*}(\gamma)^{\frac12}Z-\frac12 h^{\mathsf T}I_{\theta^*}(\gamma)h\right).
\end{split}
\end{equation*}
 It follows then from Lemmas \ref{lem:convergence_P_to_p_gamma_in_log} and \ref{lem:convergence_two_estimator} that,
\begin{equation*}
\begin{split}
\lim_{n\to\infty} \left(\sqrt{n}(\hat \theta_n-(\theta^*+h/\sqrt{n})),\ \ln\frac{\mathbb P_{\theta^*+h/\sqrt{n}}(\omega_1,\ldots,\omega_n)}{\mathbb P_{\theta^*}(\omega_1,\ldots,\omega_n)}\right)\\
\stackrel{\mathcal L-\mathbb P_{\theta^*|\gamma}}=\left(I_{\theta^*}^{-\frac12}(\gamma)Z-h,\ h^{\mathsf{T}}I_{\theta^*}(\gamma)^{\frac12}Z-\frac12 h^{\mathsf T}I_{\theta^*}(\gamma)h\right).
\end{split}
\end{equation*}
Now, from Le Cam's third Lemma we get
$$\lim_{n\to\infty} \sqrt{n}(\hat\theta_n -(\theta^*+h/\sqrt{n}))\stackrel{\mathcal L-\mathbb P_{\theta^*+h/\sqrt{n}}}=I_{\theta^*}(\gamma)^{-\frac12}Z.$$
So that, Lemma \ref{lem:conditionning_on_tail} yields the proposition.
\end{proof}

\section{Applications to Quantum non-Demolition Measurement}\label{sec:QND}

As mentioned in the Introduction, the above development is motivated by some applications in quantum physics. In particular, as we will see, the above estimation results can be applied in the context of QND measurement. For the sake of completeness we recall briefly the QND model. For a complete overview of this model we refer to \cite{bbb1,bbb2}.

Let $\{e_\alpha\ |\ \alpha\in\mathcal P\}$ and $\{\psi_j\ |\ j\in\mathcal A\}$ be orthonormal basis of respectively $\mathbb C^d$ and $\mathbb C^l$. These last spaces are endowed with their canonical Hilbert space structure. In   the context of quantum physics, these basis will be  associated  with some physical quantities. Each vector of these basis describes the physical state corresponding to an almost sure value of said physical quantities. The Hilbert space $\mathbb C^d$ describes the quantum system that one   aims  to measure indirectly.  The Hilbert space  $\mathbb C^l$ describes a probe that will be used to measure indirectly. Now, we detail the usual setup of indirect measurement. It consists in measuring something on the probe after some interaction between the system and the probe. More precisely,
the interaction is described  through a unitary operator $U$ on $\mathbb C^d\otimes \mathbb C^l$.  For QND measurements, this operator may be written as 
$$U=\sum_{\alpha\in\mathcal P}\pi_{e_\alpha}\otimes U_\alpha.$$
Here, $\pi_{e_\alpha}$ is the  projector on  the line $\mathbb C e_\alpha$ and $(U_\alpha)_{\alpha\in\mathcal P}$ are unitary operators on  $\mathbb C^l$.  $(U_\alpha)$ depends on the unknown parameters of the experiment.   The state of the system is represented by  the unit vector\footnote{Actually the state corresponds to the line given by the direction $\phi_0$. Particularly two states are equivalent if they differ only by a phase. We no longer mention it and always mean equality up to a phase when comparing two states.} $\phi_0\in\mathbb C^d$. This vector may be expanded on the first basis:
$$\phi_0=\sum_{\alpha}\langle e_\alpha,\phi_0\rangle\,e_\alpha.$$
The state of the probe is represented by a unit vector $\psi \in \mathbb  C^l$. After the interaction the joint system--probe state is  the unit vector $ U(\phi_0\otimes\psi)\in \mathbb C^d\otimes \mathbb C^l$. This vector may be expanded as,
$$U(\phi_0\otimes\psi)=\sum_{\alpha}\langle e_\alpha,\phi_0\rangle\,e_\alpha\otimes U_\alpha \psi.$$
We are now in position to see how multinomial mixtures encompass the law of sequence of measurement results in QND measurements.

To begin with, let us  assume that $\phi_0=e_\alpha$ for some $\alpha\in\mathcal P$. Then, 
$$q_0(\alpha):=|\langle e_\alpha,\phi_0\rangle|^2=1.$$ 
Further, from the definition of $U$,
$$U(\phi_0\otimes\psi)=e_\alpha\otimes U_\alpha\psi.$$
This property justifies the denomination {\it non demolition} explained in the Introduction. If the system state is $e_\alpha$ before the interaction it remains $e_\alpha$ after the interaction. Quantum mechanics tells that measuring a physical state in $\{\psi_j\ |\ j\in\mathcal A\}$ has the following probability distribution
$$\mathbb P[\mbox{\tt observing}\,\,\psi_j]=\vert \langle \psi_j,U_\alpha\psi\rangle\vert^2.$$
We set 
$$p(j\vert\alpha):=\vert \langle \psi_j,U_\alpha\psi\rangle\vert^2.$$
We will later see how  $U_\alpha$ depends on the unknown parameter $\theta$ (and will add an index $\theta$ to this notation). 
If $\phi_0\not\in\{e_\alpha\ |\ \alpha\in\mathcal P\}$, the probe measurement outcome is $\psi_j$ with probability
\begin{align*}
\pi_0(j)=\mathbb P[\mbox{\tt observing}\,\,\psi_j]&=\sum_\alpha \vert\langle e_\alpha,\phi_0\rangle\vert^2\vert \langle \psi_j,U_\alpha\psi\rangle\vert^2\\
&=\sum_{\alpha}q_0(\alpha)p(j\vert\alpha)
\end{align*}
The quantum mechanics projection postulate implies that if $\psi_j$ is the measurement outcome then the joint system--probe state  becomes 
$$\tilde\phi_1(j)=\frac{\sum_{\alpha}\langle \phi_0,\alpha\rangle\langle U_\alpha\psi,j\rangle}{\sqrt{\pi_0(j)}}\otimes \psi_j$$
Hence, the system state goes from $\phi_0$ to
\begin{equation}\label{eq:state_markov_chain}
\phi_1(j)=\frac{\sum_{\alpha}\langle \phi_0,\alpha\rangle\langle U_\alpha\psi,j\rangle}{\sqrt{\pi_0(j)}}.
\end{equation}
In others words, this is the new system state conditioned on the outcome $\psi_j$ (for the probe).  This system state update leads to the update of $q_0(\alpha)$,
$$q_1(\alpha):=|\langle e_\alpha,\phi_1\rangle|^2=q_0(\alpha)\frac{p(j\vert\alpha)}{\pi_0(j)}$$
This procedure results in the definition of random variables $\phi_1$ and $(q_1(\alpha))_{\alpha\in\mathcal P}$ whose laws are images of the law $\mathbb P[\mbox{\tt observing}\,\,\psi_j]=\pi_0(j)$. 

Now, we repeat the previous steps. Let $(X_n)$ be the resulting sequence of outcome (identifying $\psi_j$ and $j$). We have
\begin{align}
\mathbb P[X_1=j]=\pi_0(j)&=\sum_{\alpha}q_0(\alpha)p(j\vert\alpha),\quad  (j\in\mathcal A),\\
\mathbb P[X_1=j,X_2=j']&=\sum_{\alpha}q_0(\alpha)p(j\vert\alpha)p(j'\vert\alpha),\quad (j,j'\in\mathcal A), \\
\mathbb P[X_1=j_1,\ldots, X_n=j_n]&=\sum_{\alpha}q_0(\alpha)\prod_{k=1}^np(j_k\vert\alpha),\quad (n\in\mathbb N_*;\; j_1,\cdots,j_n\in\mathcal A).
\end{align} 
Using Kolomgorov's consistency Theorem 
we thus have
defined the law of the random sequence
$(X_n)$.  For $\alpha\in\mathcal P$, let $\mathbb P_\alpha$ be  the probability measure such that 
\begin{equation}
\mathbb P_\alpha[\{(j_1,\ldots, j_n,\omega):\; \omega\in \Omega\}]=\prod_{k=1}^np(j_k\vert\alpha).
\label{Koko}
\end{equation}
Then the law of $(X_n)$ is the mixture of multinomials
$\sum_\alpha q_0(\alpha)\mathbb P_\alpha$.
Let turn now to the statistical model. For any $\alpha\in\mathcal P$, the unitary operator $U_\alpha=U_\alpha(\theta)$ depends  on the unknown parameter $\theta$. Hence, we wish to study the statistical model $(\mathbb P_\theta)$ with
$$\mathbb P_\theta:=\sum_\alpha q_0(\alpha)\mathbb P_{\theta,\alpha}$$
where $\mathbb P_{\theta,\alpha}$ is defined in (\ref{Koko}) with
$$p_{\theta}(j\vert\alpha)=\vert \langle U_\alpha(\theta)\psi,j\rangle\vert^2.$$
Now, all the results developed in the last sections hold whenever the regularity assumptions are assumed directly on $(U_{\alpha}(\cdot))_{\alpha\in\mathcal P}$. 
Let us now unravel what is the Fisher information for the QND measurement model. 

\medskip
\noindent
Let $\alpha\in\mathcal P$ and $\Theta$ be some given non empty open subset of an Euclidean space. Further, let $H_\alpha(\cdot)$ be a differentiable function on $\Theta$ taking its values in  the set of self-adjoint $l\times l$ matrices. 
For $\alpha\in\mathcal P$, set $U_{\alpha}(\cdot):=\exp(-i H_\alpha(\cdot))$. Then, a direct differentiation gives,
$$\partial_{\theta_k} \ln(p_\cdot(i\vert\alpha))=2\operatorname{Im}\left(\frac{\langle\partial_{\theta_k} H_\alpha(\cdot) U_\alpha(\cdot)\psi,\psi_j\rangle}{\langle U_{\alpha}(\cdot)\psi,\psi_j\rangle}\right).$$
It follows that,
$$(I_\cdot(\alpha))_{kl}=4\sum_{j\in\mathcal A}p_\cdot(j|\alpha) \operatorname{Im}\left(\frac{\langle\partial_{\theta_k} H_\alpha(\cdot) U_\alpha(\cdot)\psi,\psi_j\rangle}{\langle U_{\alpha}(\cdot)\psi,\psi_j\rangle}\right)\operatorname{Im}\left(\frac{\langle\partial_{\theta_l} H_\alpha(\cdot) U_\alpha(\cdot)\psi,\psi_j\rangle}{\langle U_{\alpha}(\cdot)\psi,\psi_j\rangle}\right).$$
More particularly, in the relevant one dimensional case $d=1$ where $H_{\alpha}(\theta)=\theta H_\alpha$ for some self adjoint matrices $H_\alpha$,
$$\partial_\theta \ln(p_\theta(i\vert\alpha))=2\operatorname{Im}\left(\frac{\langle H_\alpha U_\alpha(\theta)\psi,\psi_j\rangle}{\langle U_{\alpha}(\theta)\psi,\psi_j\rangle}\right).$$
So that,
$$I_\theta(\alpha)=4\sum_{j\in\mathcal A}p_\theta(j|\alpha) \operatorname{Im}^2\left(\frac{\langle H_\alpha  U_\alpha(\theta)\psi,\psi_j\rangle}{\langle U_{\alpha}(\theta)\psi,\psi_j\rangle}\right).$$

\begin{figure}[h!]
\includegraphics[width=.7\textwidth]{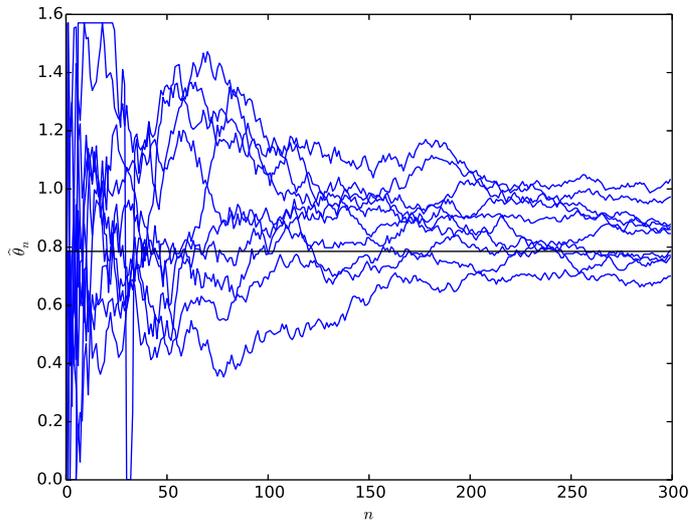}
\caption{Numerical simulation of $10$ realizations of $(\theta_n)$ for the toy model described in the text with $q(\alpha)$ proportional to $3.46^\alpha/\alpha!$ for $\alpha\in\{1,\ldots,8\}$. The black line corresponds to the target value $\theta^*=\pi/4$.\label{fig:simu}}
\end{figure}

\noindent{\bf Toy QND example:}\\
We conclude our paper  by an illustration of our results on the Haroche's group experiment \cite{guerlin} mentionned in the Introduction. In that experiment $d=8$ and $\mathcal A=\{0,1\}\times\{0,1,2,3\}$ hence $l=8$ and $j$ is defined through a bijection $j\equiv (x,a)$. The expression of $p(j\vert\alpha)$ is provided in \cite{guerlin}. Expressed in  our notations we have here,\footnote{Remark that in \cite{guerlin} $\alpha=0,\ldots, 7$ so Assumption {\bf ID} is not verified for $\alpha=0$. More precisely $\theta_4$ may not be identifiable. Though for the ideal value of $\theta=\pi/4$, $\alpha=0$ is equivalent to $\alpha=8$ so we use this value. Hence, we identify the zero photon state with the eight one at the opposite of what is done in \cite{guerlin}.}
$$p_\theta(x,a | \alpha)=(\theta_5+\theta_6 \cos( \alpha\theta_4 + \theta_a+x\pi))/8.$$
Hence,  $\dim \Theta=7$. The ideal values of the parameters are $\theta_5=\theta_6=1$, $\theta_4=\pi/4$ and $\theta_a=(2-a)\pi/4$. Experimentally some imperfections imply that $\theta_6$ is smaller than $1$. It is close to $\theta=0.674\pm0.004$. Besides this limitation in \cite{guerlin} the authors find parameters close to their target values using a best fit to the empirical distribution. Setting $\theta_6=0.674$, it is easy to check that {\bf ID} holds for the ideal parameters and in a small enough but sufficiently large neighborhood of the true parameters. Moreover, all the functions $\theta\mapsto p_{\theta}(x,a|\alpha)$ are entire analytic. 
\noindent
We limit ourselves to the estimation of $\theta_4$ the other parameters are fixed to their true values except $\theta_6$. We set $\theta_6=0.674$. Further, we take $\Theta=[\pi/8,3\pi/8]$ and,
$$p_\theta(x,a|\alpha):=(1+0.674\cos(\alpha\theta +(2-a)\pi/4+x\pi))/8.$$
It follows that
$$I_\theta(\alpha)=\sum_{a=0}^3\alpha^2\frac{0.674^2}{4}\frac{\sin^2(\alpha\theta +(2-a)\pi/4)}{1-0.674^2\cos^2(\alpha\theta +(2-a)\pi/4)}.$$
\noindent
Fisher information is not singular at $\theta=\pi/4$. Figure \ref{fig:simu} depicts some simulations of $(\hat \theta_n)$ for this toy model.
\bigskip

\noindent \textbf{Acknowledgments} T.B and C.P are supported by ANR project StoQ ANR-14-CE25-0003-01 and CNRS InFIniTi project MISTEQ. The research of T.B. has been supported by ANR-11-LABX-0040-CIMI within the program ANR-11-IDEX-0002-02.

\appendix
\section{Uniform convergence of logarithm of sums of uniformly convergent sequences}

\begin{lemma}\label{lem:log_sum_unfi_conv}
Let $(a_n)$ and $(b_n)$ be two sequence of strictly positive functions with common domain of definition such that there exists two functions $\ell_a$ and $\ell_b$ such that,
$$\lim_{n\to\infty}\left\|\tfrac1n\ln a_n- \ell_a\right\|_{\infty}=0\quad\text{and}\quad\lim_{n\to\infty}\left\|\tfrac1n\ln b_n- \ell_b\right\|_{\infty}=0$$
with $\|\cdot\|_\infty$ the sup norm. Then,
$$\lim_{n\to\infty}\left\|\tfrac1n \ln(a_n+b_n)-\max(\ell_a,\ell_b)\right\|_\infty=0.$$
\end{lemma}
\begin{proof}
Let $D$  be the common domain of definition of the functions.
Assumptions imply there exists a sequence $(\epsilon_n)$ of strictly positive numbers such that $\lim_{n\to\infty}\epsilon_n=0$ and
$$\max\left(\left\|\tfrac1n\ln a_n- \ell_a\right\|_{\infty},\left\|\tfrac1n\ln b_n- \ell_b\right\|_{\infty}\right)\leq \epsilon_n.$$
Since the function $\ln$ is non decreasing, we deduce first that for any $x\in D$,
$$\tfrac1n \ln(a_n(x)+b_n(x))-\max(\ell_a(x),\ell_b(x))\leq \epsilon_n+\tfrac1n \ln 2.$$
Second, let $c(x)$ be such that $\ell_c=\max(\ell_a,\ell_b)$. Then for any $x\in D$,
$$-\epsilon_n\leq \tfrac1n \ln c_n(x) - \ell_c(x)\leq \tfrac1n \ln(a_n(x)+b_n(x)) - \max(\ell_a(x),\ell_b(x))$$
and the Lemma holds.
\end{proof}
\end{document}